\documentclass[hidelinks, 12pt]{article}
\usepackage[mathscr]{eucal}
\usepackage{amssymb,latexsym,amsthm,pifont,graphicx,anysize,multicol}
\usepackage{verbatim}
\usepackage{graphicx}
\usepackage{color}
\usepackage{amsmath}
\usepackage{amsthm}
\usepackage{enumerate}
\usepackage{framed}
\usepackage{authblk}
\usepackage{setspace}
\usepackage{hyperref}
\usepackage{cancel}

\usepackage{bm}
\usepackage{tocvsec2}
\usepackage{comment} 
\usepackage{setspace}
\usepackage{hyperref}
\usepackage{epigraph}
\usepackage[margin=1.5cm]{caption}

\definecolor{dgray}{RGB}{90,90,90}
\definecolor{gray}{RGB}{120,120,120}
\definecolor{lgray}{RGB}{150,150,150}
\definecolor{purple}{RGB}{150,0,150}
\definecolor{magenta}{RGB}{250,0,150}

\numberwithin{equation}{section}

\newtheorem{thm}{Theorem}[section]
\newtheorem{lem}[thm]{Lemma}
\newtheorem{Def}[thm]{Definition}

\newtheorem{prop}[thm]{Proposition}
\newtheorem{cor}[thm]{Corollary}
\newtheorem{conj}[thm]{Conjecture}

\renewcommand\S{\Sigma}
\newcommand\s{\sigma}
\renewcommand\d{\partial}

\newcommand\D{\nabla}
\newcommand\e{\epsilon}
\renewcommand\b{\beta}

\newcommand\ric{{\rm Ric}}

\newcommand\g{\gamma}
\newcommand\8{\infty}
\renewcommand\a{\alpha}

\newcommand{\field}[1]{\mathbb{#1}}

\DeclareFontFamily{OT1}{rsfs}{}
\DeclareFontShape{OT1}{rsfs}{m}{n}{ <-7> rsfs5 <7-10> rsfs7 <10->
rsfs10}{} \DeclareMathAlphabet{\mycal}{OT1}{rsfs}{m}{n}

\newcommand\vs{\vspace}

\newcommand\beq{\begin{equation}}
\newcommand\eeq{\end{equation}}
\newcommand\ben{\begin{enumerate}}
\newcommand\een{\end{enumerate}}
\newcommand\bit{\begin{itemize}}
\newcommand\eit{\end{itemize}}

\newcounter{mnotecount}[section]

\title{
Hausdorff closed limits and rigidity in \\ Lorentzian geometry}

\author[*]{Gregory J. Galloway}
\author[$\dag$]{Carlos Vega}

\affil[*]{\small Department of Mathematics, 

University of Miami, Coral Gables, FL }

\affil[$\dag$]{Department of Mathematics,

Saint Louis University, St. Louis, MO}

\begin{document}
\date{}

\newpage
\thispagestyle{empty}
\maketitle

\begin{center}
{\it Dedicated to Robert Bartnik on the occasion of his 60th birthday}
\end{center}
\centerline{ \; \\}
\vspace{2pc}

\begin{abstract} 
We begin with a basic exploration of the (point-set topological) notion of Hausdorff closed limits in the spacetime setting. Specifically, we show that this notion of limit is well suited to sequences of achronal sets, and use this to generalize the `achronal limits' introduced in \cite{horo1}. This, in turn, allows for a broad generalization of the notion of Lorentzian horosphere introduced in \cite{horo1}. We prove a new rigidity result for such horospheres, which in a sense encodes various spacetime splitting results, including the basic Lorentzian splitting theorem. We use this to give a partial proof of the Bartnik splitting conjecture, under a new condition involving past and future Cauchy horospheres, which is weaker than those considered in \cite{GalBanach} and \cite{horo1}. We close with some observations on spacetimes with spacelike causal boundary, including a rigidity result in the positive cosmological constant case.
\end{abstract}

\newpage

%
%
\renewcommand\contentsname{}
\tableofcontents

\section{Introduction} 

In the spirit of the classical horospheres of hyperbolic geometry, the authors introduced a natural geometric and causal theoretic notion of horosphere in Lorentzian geometry in \cite{horo1}. By virtue of this approach, many of the technical analytic difficulties in dealing with conventional Lorentzian horospheres (associated to timelike rays via Lorentzian Busemann functions) is circumvented. The approach in \cite{horo1} also allowed for more general types of horospheres, including a new `Cauchy horosphere'. In the present paper we consider a very broad generalization of the definition of horosphere in \cite{horo1} based on Hausdorff closed limits.  As noted in \cite{BEE}, the important limit curve concept in Lorentzian geometry can be described in terms of such limits.  Somewhat in analogy, here we define a Lorentzian horosphere as the Hausdorff closed limit of a certain class of  Lorentzian spheres, which are in particular achronal boundaries; see Figure \ref{Minkrayhoro}. 

In Section \ref{secHausdorff} we review the definition of Hausdorff closed limits and establish some fundamental properties. In particular, we show that these limits preserve achronality and edgelessness, and further show that the Hausdorff closed limit of achronal boundaries is an achronal boundary itself. In Section~\ref{secspheresandhoros} we define a horosphere to be the Hausdorff closed limit of Lorentzian spheres, with `causally complete' centers, as the radii tend to infinity. This drops the monotonicity requirement used in \cite{horo1}, and the horospheres defined in \cite{horo1}, including the Ray horospheres and Cauchy horospheres, now become a special subclass. In Section \ref{sechorostructure} we present a very general splitting theorem for past and future horospheres that meet in a `noncrossing manner', which supersedes many known Lorentzian splitting results. In Section \ref{secapps} we discuss various applications of this horosphere splitting theorem to, e.g., the Lorentzian splitting theorem and the Bartnik splitting conjecture, as well as some rigidity results for spacetimes with spacelike (past or future) causal boundary.

\begin{figure}
\begin{center}
\def\svgwidth{12.1cm} 
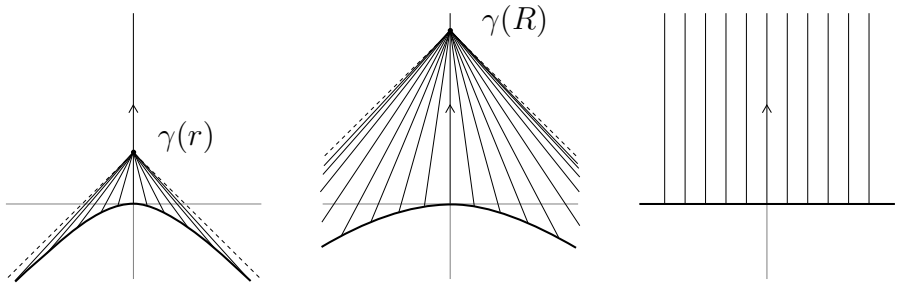
\caption[...]{The prototypical Lorentzian horosphere, from a future timelike ray $\g$, is shown here in Minkowski space. The conventional construction relies on the Busemann function associated to $\g$. The approach here is to define the horosphere directly as the Hausdorff closed limit of the sequence of past spheres from $\g$.}
\label{Minkrayhoro}
\end{center}
\end{figure}

\smallskip
\noindent
{\it Acknowledgments.}  Work on this paper was initiated during the Fall 2013 Program on Mathematical General Relativity at the Mathematical Sciences Research Institute (MSRI).  Further work was carried out at the Workshop on Geometric Analysis and General Relativity at the Banff International Research Station (BIRS) in July, 2016.  The authors would like to express their gratitude to both MSRI and BIRS for their support and for providing, in each case, a highly conducive research environment.  The work of GJG was also partially supported  by NSF grant DMS-1313724.

\section{Lorentzian Preliminaries} 

We begin with some brief Lorentzian preliminaries. 
For further background, we note the standard references \cite{ON}, \cite{BEE}, \cite{Wald}, \cite{HE}. This section also serves to set a few conventions, though we note that these remain unchanged from \cite{horo1}.

Throughout the following, $M = (M^{n+1}, g)$ will denote a spacetime, i.e., a connected, time-oriented Lorentzian manifold, of dimension $n +1 \ge 2$. We take the Lorentzian metric $g$ to be smooth. A vector $X \in TM$ will be called \emph{timelike} if $g(X, X) < 0$, \emph{null} if $g(X,X) = 0$, and \emph{spacelike} if $g(X,X) > 0$. We say $X$ is \emph{causal} if $X$ is either timelike or null.  

The time-orientation of $M$ implies that any nontrivial causal vector points either to the `future' or to the `past'. By a \emph{future causal curve}, we mean a piecewise-smooth curve $\a$, with $\a'$ always future causal, (including any one-sided tangents). Past causal curves are defined time-dually, and future/past timelike/null curves are defined analogously. By a `causal curve' we will always mean either a future causal curve or a past causal curve, and similarly for timelike and null curves.

If there is a future causal curve from $p \in M$ to $q \in M$, we write $p \le q$, or equivalently $q \in J^+(p)$, or $p \in J^-(q)$. If there is a future timelike curve from $p$ to $q$, we write $p \ll q$, or equivalently $q \in I^+(p)$, or $p \in I^-(q)$. More generally, we define the causal future $J^+(S)$ of any subset $S \subset M$ to be the set of points $q \in M$ which can be reached by a future causal curve starting from $S$. The sets $J^-(S)$, and $I^+(S)$ and $I^-(S)$, are defined analogously. 

Finally, we note that a spacetime is \emph{globally hyperbolic} if the set of all `timelike diamonds' $I^+(p) \cap I^-(q)$ forms a basis for the manifold topology, and all `causal diamonds' $J^+(p) \cap J^-(q)$ are compact. Because of its relationship to Lorentzian distance, we work exclusively in the globally hyperbolic setting from Section \ref{secspheresandhoros} on. All of Section \ref{secachlimits}, however, applies to general spacetimes. For further background on global hyperbolicity, and causal theory in general, we defer again to the references above.
 
\section{Achronal Limits} \label{secachlimits}

A subset $A \subset M$ is called \emph{achronal} if no two points in $A$ are joined by a timelike curve, i.e., $I^+(A) \cap A = \emptyset$. It is a basic causal theoretic fact that any achronal set without `edge' points is a $C^0$ hypersurface in $M$. A special case of such a set is that of an `achronal boundary', i.e., any nonempty set of the form $A = \partial I^\pm(S)$. 

In Section \ref{secABs}, we first treat some of the basic theory of achronal sets, and of achronal boundaries specifically, as studied by Penrose in \cite{Penrose}. In Section \ref{secHausdorff}, we then use the notion of Hausdorff closed limits to broadly generalize the `achronal limits' introduced in \cite{horo1}.

\subsection{Achronal Sets} \label{secABs}

To define the `edge' of an achronal set, we must first recall the notion of local or relative causality. Let $U \subset M$ be any open neighborhood, and let $p \in U$. By $I^+(p,U)$ we mean the timelike future of $p$ within the (sub)spacetime $U$. That is, $q \in I^+(p,U)$ iff there is a future timelike curve from $p$ to $q$ which lies completely within $U$. $I^-(p,U)$ is defined time-dually.

Now let $A \subset M$ be any achronal set. The \emph{edge} of $A$ is defined to be the set of points $p \in \overline{A}$ such that every neighborhood $U$ of $p$ contains a timelike curve from $I^-(p,U)$ to $I^+(p,U)$ which does not meet $A$. We say $A$ is \emph{edgeless} if $\textrm{edge}(A) = \emptyset$. The following is one of the fundamental consequences of achronality.

\begin{prop} [See \cite{ON}] Let $A$ be any nonempty achronal set. Then $A$ is a (topologically) closed $C^0$ hypersurface iff $A$ is edgeless. 
\end{prop}

We now proceed to the special case of achronal boundaries. As in \cite{Penrose}, we say a subset $P \subset M$ is a \emph{past set} if it is the timelike past of a set, i.e., $P = I^-(S)$, for some $S \subset M$. It follows that $P$ is a past set iff $P = I^-(P)$. Future sets are defined time-dually. The nonempty boundary of a past or future set is called an \emph{achronal boundary}. Hence, an achronal boundary is a set of the form $\emptyset \ne A = \d I^\pm(S)$. 

\begin{prop} [\cite{Penrose}] \label{ABs} Let $A$ be an achronal boundary. Then $A$ is achronal and edgeless, and hence a closed $C^0$ hypersurface. Moreover, there is a unique past set $P$ such that $A = \partial P$, and a unique future set $F$ such that $A = \d F$, and this triple forms a disjoint partition, $M = P \cup A \cup F$. It follows that $I^-(A) \subset P$ and $I^+(A) \subset F$. 
\end{prop}

Hence, if $A$ is an achronal boundary, then any future timelike curve from $I^-(A)$ to $I^+(A)$ must pass through $A$. While this fails in general if $A$ is only taken achronal and edgeless, the following result says, in effect, that this does hold locally.

\begin{lem} \label{diamondslicing} Let $A$ be an achronal and edgeless subset of a spacetime $(M,g)$.   Let $U$ be a convex normal neighborhood of $M$, and let $N$ be a globally hyperbolic sub-spacetime of $(U,g|_U)$.  If $A$ enters a timelike diamond $I^+(x,N) \cap I^-(y,N)$, then any future timelike curve from $x$ to $y $ in $N$ must meet $A$. 
\end{lem}

\begin{proof} Consider $A_0 := A \cap J^+(x, N) \cap J^-(y, N)$. Note that $x \in I^-(A_0, N)$ and $y \in I^+(A_0, N)$. Hence, letting $\b$ be any future timelike curve in $N$ from $x$ to $y$, then $\b$ must meet $\partial I^-(A_0, N)$ at some point $z_0$. Since $A$ is closed, and $N$ is globally hyperbolic, it follows that $A_0$ is compact, and hence $J^-(A_0, N)$ is closed in $N$. Thus $z_0 \in J^-(A_0, N)$. If $z_0 \in A_0$, we are done. Suppose then that $z_0 \not \in A_0$. Then by standard causal theory, there is a future null geodesic $\eta$ in $N$ from $z_0$ to $a_0 \in A_0$, with $\eta \subset \partial I^-(A_0, N)$, and $\eta \cap A_0 = \{a_0\}$. Since $a_0 \in J^-(y, N)$ and $N$ is globally hyperbolic, there is a future causal geodesic $\zeta$ in $N$ from $a_0$ to $y$. Hence, the concatenation $\eta + \zeta$ gives a future causal curve in $N$ from $z_0$ to $y$. But since $z_0 \in I^-(y, U)$ and $U$ is convex, the {\it unique} geodesic joining $z_0$ and $y$ in $U$ is timelike. It follows that $\eta$ and $\zeta$ must form a `corner' at $a_0$, and hence that every point of $\zeta \setminus \{a_0\}$ is in the timelike future of every point of $\eta \setminus \{a_0\}$. Moreover, since $z_0 \in \partial I^-(A_0, N)$, it follows that $\zeta \cap A_0 = \{a_0\}$. But then $\eta \cap A = \{a_0\} = \zeta \cap A$ implies that $a_0$ is an edge point of $A$, a contradiction. Hence $\beta$ does in fact meet $A$ at $z_0 \in A_0 \subset A$.
\end{proof}

\subsection{Achronal Limits} \label{secHausdorff}

In \cite{horo1}, a natural notion of `achronal limit' was defined for sequences of achronal boundaries exhibiting a basic kind of monotonicity. The results of this section broadly generalize such limits, using so-called `Hausdorff closed limits'. In particular, we establish the following facts:

\begin{thm} \label{achlimitsresults} Let $\{A_k\}$ be any sequence of subsets with Hausdorff closed limit, 
$$A_\infty = \lim \{A_k\}$$
If each $A_k$ is achronal, then so is $A_\infty$. If further each $A_k$ is edgeless, then so is $A_\infty$. Finally, if each $A_k$ is an achronal boundary, then so too is $A_\infty$. 
\end{thm}

Theorem \ref{achlimitsresults} thus demonstrates that all of the basic properties of achronal sets are preserved under Hausdorff closed limits. In light of these results, if $\{A_k\}$ is any sequence of achronal subsets, with Hausdorff closed limit $A_\infty = \lim \{A_k\}$, we will call $A_\infty$ the \emph{achronal limit} of $\{A_k\}$. That such limits do indeed generalize those in \cite{horo1} follows immediately from Proposition 2.5 in \cite{horo1} and Lemma \ref{Hausdorffseq} below. (Further discussion of this point is included at the end of this subsection.)

As we will see, the first statement in Theorem \ref{achlimitsresults} follows quite easily, while the next two are somewhat more subtle. The complete proof will be carried out in stages, culminating in Theorems \ref{achlimitsedgeless} and \ref{achlimitsequiv} below. As an immediate application, these results will be used in Section \ref{sechoros} to generalize the horospheres defined in \cite{horo1}.

We begin by recalling the following definitions, introduced by Hausdorff in \cite{HausdorffST}, and used, for example, in \cite{Busemann}, \cite{BEE}, \cite{Papadopoulos}.

\begin{Def} [Hausdorff Closed Limits, \cite{HausdorffST}] \label{defHausdorff} Let $\{S_k\}$ be a sequence of subsets of a topological space $\mathcal{M}$. The \emph{Hausdorff upper} and \emph{lower limits} of $\{S_k\}$ are defined, respectively, by

\vspace{1pc}
\centerline{\, $S^\mathrm{up}_\infty = \overline{\lim} \{S_k\} = \{p :$ \text{each neighborhood of $p$ meets infinitely many $S_k$'s}$\}$ \; \; }

\vspace{1pc}

\centerline{\, $S^\mathrm{low}_\infty = \underline{\lim} \{ S_k\}  = \{p  :$ \text{each neighborhood of $p$ misses only finitely many $S_k$'s}$\}$}

\vspace{1pc}
\noindent
Hence, in general, $\underline{\lim} \{ S_k\} \subset \overline{\lim} \{S_k\}$. In the case of equality, the common limit is called the \emph{Hausdorff closed limit} of $\{S_k\}$, which we denote by $S_\infty = \lim \{ S_k \}$. 
\end{Def}

\medskip
It is straightforward to check that $S^\mathrm{up}_\infty$ and $S^\mathrm{low}_\infty$ are closed. Hence, when it exists, the Hausdorff closed limit $S_\infty$ is indeed closed. In a metric space, this notion of limit is closely related to convergence of subsets with respect to the Hausdorff distance; see \cite{Papadopoulos} for some basic discussion. Moreover, the following characterizations are easily verified:

\begin{lem} \label{Hausdorffseq} Let $\{S_k\}$ be a sequence of subsets of a metric space $\mathcal{M}$.
\ben
\item [(1)] $S^\mathrm{up}_\infty $ is precisely the set of limit points of sequences $s_k \in S_k$. 
\item [(2)] $S^\mathrm{low}_\infty $ is precisely the set of limits of sequences $s_k \in S_k$. 
\een
In particular, if $S_\8$ exists, then any limit point of a sequence $x_k \in S_k$ is in $S_\8$, and every point in $S_\8$ is the limit of some (convergent) sequence $y_k \in S_k$.
\end{lem}

The following implies that achronality is preserved under Hausdorff closed limits.

\begin{lem} \label{achronallower} Let $\{A_k\}$ be any sequence of achronal subsets of a spacetime $M$. Then the Hausdorff lower limit, $A^\textrm{low}_\infty$, is achronal.   
\end{lem}

\begin{proof} Suppose to the contrary that there is a future timelike curve $\a : [0, b] \to M$ from $\a(0) \in A^\textrm{low}_\infty$ to $\a(b) \in A^\textrm{low}_\infty$. For $0 < \e$ small, let $U := I^-(\a(\e))$ and $V := I^+(\a(b-\e))$. Hence, $U$ is an open neighborhood of $\a(0)$, and $V$ is an open neighborhood of $\a(b)$, and every point in $U$ is timelike related to every point in $V$. Note that $\a(0)$ is the limit of a sequence $x_k \in A_k$, and $\a(b)$ is the limit of a sequence $y_k \in A_k$. Hence, for all sufficiently large $k$, $A_k$ must enter both $U$ and $V$. But this violates the achronality of $A_k$. 
\end{proof}

\begin{figure} [h]
\begin{center}
\def\svgwidth{13.5cm} 
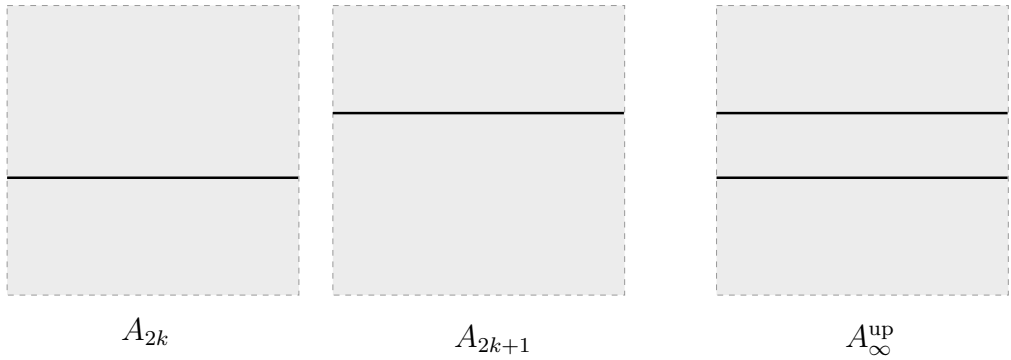
\caption[...]{The upper limit of achronal sets may not be achronal.
\label{nonachronalupper}} 
\end{center}
\end{figure}

\begin{figure} [h]
\begin{center}
\def\svgwidth{13.5cm} 
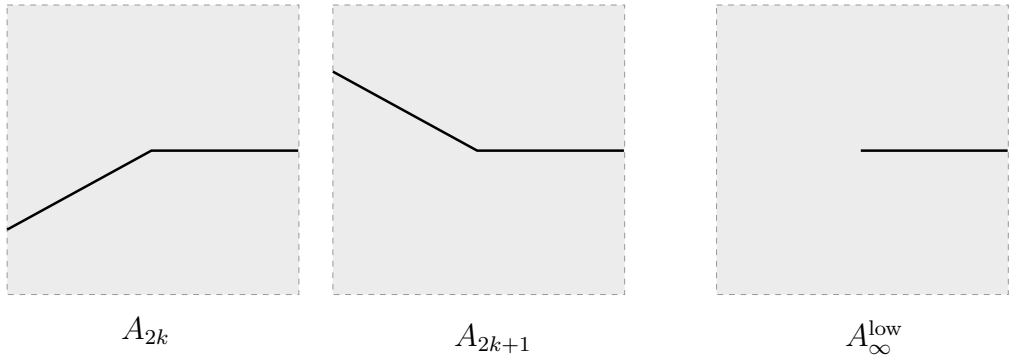
\caption[...]{The lower limit of achronal sets may develop an edge.
\label{nonedgelesslower}} 
\end{center}
\end{figure}

Figure \ref{nonachronalupper} shows how Lemma \ref{achronallower} can fail for upper limits. On the other hand, Figure \ref{nonedgelesslower} shows how a lower limit can develop an edge.
By contrast, upper limits do not develop edge points in this way, as shown in the next lemma.

\begin{lem} \label{edgelessupper2} Let $\{A_k\}$ be a sequence of edgeless achronal sets, such that the Hausdorff upper limit $A^\textrm{up}_\infty$ is achronal. Then $A^\textrm{up}_\infty$ is also edgeless.
\end{lem}

\begin{proof} Since $A^\textrm{up}_\infty$ is closed, note that $\textrm{edge}(A^\textrm{up}_\infty) \subset A^\textrm{up}_\infty$. Fix any (candidate edge point) $p \in A^\textrm{up}_\infty$. Let $N$ be a neighborhood of $p$ which is globally hyperbolic as a (sub)spacetime and which is contained in a convex normal neighborhood $U$ of $p$. (For example, one may take a spacelike hypersurface through $p$ which is acausal within $U$, and take $N$ to be its domain of dependence relative to $U$.)  Fix $x \in I^-(p,N)$ and $y \in I^+(p,N)$, and let $\b$ be a future timelike curve from $x$ to $y$ in $N$. It follows that we can find a (sub)sequence $a_j \in A_{k_j}$, with $a_j \to p$, such that $a_j \in I^+(x,N) \cap I^-(y,N)$, for all $j$. But then by Lemma \ref{diamondslicing}, $\b$ must meet each $A_{k_j}$, and hence also $A^\textrm{up}_\infty$. Hence, $p \not \in \textrm{edge}(A^\textrm{up}_\infty)$.
\end{proof}

Combining Lemmas \ref{achronallower} and \ref{edgelessupper2} gives the first half of Theorem \ref{achlimitsresults}. In particular, we note the following:

\begin{thm} \label{achlimitsedgeless} The achronal limit of a sequence of edgeless achronal hypersurfaces is itself an edgeless achronal hypersurface.
\end{thm}

We now turn to the final statement in Theorem \ref{achlimitsresults}, concerning limits of achronal boundaries. This will be established in two main steps, Propositions \ref{Hausdorffpasts} and \ref{achlimit1} below. We first note the following basic facts, (with proofs left to the reader).

\begin{lem} \label{pastconvex} Let $Q \subset M$ be an arbitrary subset. 
\ben
\item [(1)] In general, $\textrm{int} \, (Q) \subset I^-(Q)$.
\item [(2)] If $I^-(Q) \subset Q$, then $\textrm{int} \, (Q) = I^-(Q)$ and $\partial Q = \partial I^-(Q)$. 
\een
\end{lem}

The following may itself be viewed as a generalization of the achronal limits in \cite{horo1}, where the assumption of monotonicity of $\{P_k\}$ is relaxed to the existence of $\lim\{P_k\}$.

\begin{prop} \label{Hausdorffpasts} Let $\{A_k\}$ be a sequence of achronal boundaries, with associated past sets $\{P_k\}$ as in Proposition \ref{ABs}, so that $A_k = \partial P_k$. If the pasts have a Hausdorff closed limit, $\Pi_\infty := \lim \{P_k\}$, then we have the following:
\ben
\item [(1)] $\mathrm{int} \, ( \Pi_\infty) =: P_\infty$ is a past set. In particular, $P_\infty = I^-(\Pi_\infty)$.
\item [(2)] $\partial \, (\Pi_\infty)$ is the Hausdorff closed limit of $\{A_k\}$. Hence, $A_\infty = \partial P_\infty = \partial (\Pi_\infty)$.
\een
\end{prop}

\begin{proof} (1) The statement holds trivially if $\Pi_\infty = \emptyset$. Otherwise, fix $x \in \Pi_\infty$. Then we can find a sequence $x_k \in P_k$, with $x_k \to x$. Fix any $y \in I^-(x)$. Then for all large $k$, we have $y \in I^-(x_k)$, and hence $y \in P_k$. So $y \in \Pi_\infty$. This shows $I^-(\Pi_\infty) \subset \Pi_\infty$. Hence, as in Lemma \ref{pastconvex}, we have $P_\infty = \textrm{int} \, (\Pi_\infty) = I^-(\Pi_\infty)$.

(2) Suppose first that $\partial \Pi_\infty \ne \emptyset$. We will show that $\partial \Pi_\infty \subset A^\textrm{low}_\infty$, and $A^\textrm{up}_\infty \subset \partial \Pi_\infty$. First fix $x \in \d \Pi_\infty$, and let $W$ be any neighborhood of $x$. Let $U$ be a connected open neighborhood of $x$, with $U \subset W$. Since Hausdorff closed limits are closed, we have $x \in \Pi_\infty$. Hence, we can find a sequence $x_k \in P_k$ with $x_k \to x$. In particular, there is an index $k_U \in \field{N}$ such that $x_k \in P_k \cap U$ for all $k \ge k_U$. Suppose that $U \subset P_j$ for infinitely many $j$. But this implies $U \subset \Pi_\infty$, and hence $x \in \mathrm{int \;} \Pi_\infty$, contradicting $x \in \d \Pi_\infty$. Thus, there is an index $j_U \in \field{N}$ such that, for all $j \ge j_U$, we can find a point $y_j \in U \cap (A_j \cup F_j)$. Let $\ell_U = \max \{k_U, j_U\}$. Then for all $\ell \ge \ell_U$, we have points $x_\ell \in U \cap P_\ell$ and $y_\ell \in U \cap (A_\ell \cup F_\ell)$, and, since we took $U$ connected, a point $a_\ell \in U \cap A_\ell$. Hence, $U$ meets all but possibly finitely many of the $A_k$'s, and since $U \subset W$, so does $W$. This shows $\d \Pi_\infty \subset A^\textrm{low}_\infty$. 

Now fix $a \in A^\textrm{up}_\infty$. Then there is a subsequence $a_{k_j} \in A_{k_j}$ with $a_{k_j} \to a$. Hence any neighborhood $W$ of $a$ meets $P_{k_j}$ and also $F_{k_j}$ for all sufficiently large $j$. Consequently, $a$ is realizable as a limit point of a sequence $p_k \in P_k$, which means $a \in \Pi_\infty$. But also $a$ is realizable as a limit point of a sequence $f_k \in F_k$. Suppose $a \in \textrm{int} (\Pi_\infty)$. Then there is a neighborhood $V$ of $a$ contained in $\Pi_\infty$. Choose points $b,c \in V$ with $a \ll b \ll c$. Since $a$ is a limit point of a sequence $f_k \in F_k$, $b$ meets infinitely many $F_k$. But then $I^+(b)$ is a neighborhood of $c$ which is contained in infinitely many $F_k$, and hence must miss the infinitely many corresponding $P_k$. Consequently, $c \not \in \Pi_\infty$, a contradiction. Hence, $a \in \d \Pi_\infty$. This shows $A^{\textrm{up}}_\infty \subset \partial \Pi_\infty$. Thus, we have $\d \Pi_\infty \subset A^\textrm{low}_\infty \subset A^\textrm{up}_\infty \subset \d \Pi_\infty$. So $A_\infty$ exists, and $A_\infty = \d \Pi_\infty = \partial P_\infty$, the last equality following as in Lemma \ref{pastconvex}.

To finish the proof, it remains to consider the case that $\partial \Pi_\infty = \emptyset$, that is, either $\Pi_\infty = \emptyset$ or $\Pi_\infty = M$. In either case, it suffices to show $A^\textrm{up}_\infty = \emptyset$. The details are left to the reader.
\end{proof}

We now establish a converse of Proposition \ref{Hausdorffpasts}.

\begin{prop} \label{achlimit1} Suppose that $\{A_k\}$ is a sequence of achronal boundaries with Hausdorff closed limit, $A_\infty = \lim \{A_k\}$. Then the sequence of associated pasts $\{P_k\}$ also has a Hausdorff closed limit, $\Pi_\infty = \lim \{P_k\}$. 
\end{prop}

\begin{proof} Suppose not. Then it follows that there must be a point $x \in M$ such that $x \in P_k$ for infinitely many $k$, and also $x \in F_k$ for infinitely many $k$, with $x \not \in A_\infty$. Suppose that $x \in I^+(A_\infty)$. Hence, there is some $a \in A_\infty$ with $a \in I^-(x)$. Let $a_k \in A_k$ be a sequence with $a_k \to a$. Then for all sufficiently large $k$, we have $a_k \in I^-(x) \cap A_k$, and hence $I^-(x) \cap \partial F_k \ne \emptyset$. But this implies $x \in F_k$ for all large $k$, which is a contradiction. Hence $x \not \in I^+(A_\infty)$, and similarly $x \not \in I^-(A_\infty)$. Let $\a : [0, b] \to M$ be any continuous path from $\a(0) = x$ to $\a(b) \in A_\infty$, with $\a(s) \not \in A_\infty$ for all $s \in [0,b)$. Note that $A_\infty$ is achronal and edgeless, by Theorem \ref{achlimitsedgeless}. Then there must be a first parameter time $0 \le s_0 < b$ such that $\a(s_0) \in \partial I^+(A_\infty) \cup \partial I^-(A_\infty)$. It suffices to consider the case $\a(s_0) \in \partial I^+(A_\infty)$. Let $\b : [s_0, s_1] \to M$ be a future timelike curve from $\a(s_0) = \b(s_0)$. It follows that $\b(s_1) \in I^+(A_\infty)$, and hence $\b(s_1) \in F_k$ for all large $k$. Consider the path $\s : [0,s_1] \to M$ defined by $\s(s) = \a(s)$ for $0 \le s \le s_0$, and $\s(s) = \b(s)$ for $s_0 < s \le s_1$. Then, for infinitely many $k$, $\s$ is a continuous path from $P_k$ to $F_k$. It follows that $\s$ meets $A_k$ for infinitely many $k$, and hence that $\s$ meets $A_\infty$. But this is a contradiction. 
\end{proof}

Combining Propositions \ref{Hausdorffpasts} and \ref{achlimit1}, and their time-duals, we have the following:

\begin{thm} \label{achlimitsequiv} Let $\{A_k\}$ be a sequence of achronal boundaries with associated past and future sets, $\{P_k\}$ and $\{F_k\}$. Then the following are equivalent:
\ben
\item [(1)] $\{P_k\}$ has a Hausdorff closed limit, $\Pi_\infty = \lim \{P_k\}$. 
\item [(2)] $\{A_k\}$ has a Hausdorff closed limit, $A_\infty = \lim \{A_k\}$. 
\item [(3)] $\{F_k\}$ has a Hausdorff closed limit, $\Phi_\infty = \lim \{F_k\}$. 
\een
When any of the above conditions hold, let $P_\infty := \mathrm{int} \, (\Pi_\infty)$ and $F_\infty := \mathrm{int} \, (\Phi_\infty)$. Then $P_\infty$ is a past set, $F_\infty$ is a future set, and $M = P_\infty \cup A_\infty \cup F_\infty$, with $\partial P_\infty = A_\infty = \partial F_\infty$. 
\end{thm}

In particular, we note:

\begin{cor} The achronal limit of a sequence of achronal boundaries is itself an achronal boundary.
\end{cor}

We close this section by formalizing the observation that Hausdorff closed limits do indeed generalize the `achronal limits' originally defined in \cite{horo1}. Consider a sequence of achronal boundaries $\{A_k\}$, with associated pasts $\{P_k\}$, and futures $\{F_k\}$ (as per Proposition \ref{ABs}). As in \cite{horo1}, we say $\{P_k\}$ is \emph{increasing} if $P_k \subset P_{k+1}$ for all $k$, or \emph{decreasing} if $P_{k+1} \subset P_k$ for all $k$. It follows that the $F_k$'s are increasing iff the $P_k$'s are decreasing, and vice versa. We say a sequence of achronal boundaries $\{A_k\}$ is \emph{monotonic} if the pasts $P_k$ are monotonic, i.e., either increasing or decreasing. The following is an immediate consequence of Proposition 2.5 in \cite{horo1} and Lemma \ref{Hausdorffseq} above:

\begin{cor} \label{oldachlimits} Let $\{A_k\}$ be a sequence of achronal boundaries, with associated pasts $\{P_k\}$ and futures $\{F_k\}$, as in Proposition \ref{ABs}. 
\ben
\item [(1)] If $\{P_k\}$ is increasing, then the Hausdorff closed limit $\lim \{A_k\}$ exists and
\beq
\lim \{A_k\} = \partial \bigg( \bigcup_k P_k   \bigg) \nonumber \label{futureachlimit}
\eeq
\item [(2)] If $\{P_k\}$ is decreasing, and hence $\{F_k\}$ increasing, then the Hausdorff closed limit $\lim \{A_k\}$ exists and we have:
\beq
\lim \{A_k\} = \partial \bigg( \bigcup_k F_k   \bigg) \nonumber \label{pastachlimit}
\eeq
\een
\end{cor}

\section{Horospheres} \label{sechoros}

We now use the results of Section \ref{secachlimits} to generalize the notion of Lorentzian horosphere defined in \cite{horo1}.  In addition, we establish a new splitting result for such horospheres in Section \ref{sechorostructure}, (Theorem \ref{splitlem}), generalizing those in \cite{horo1}, and in a sense encoding various other spacetime splitting results, including the basic Lorentzian splitting theorem. We begin in Section \ref{secdistance} with a brief review of standard material on the Lorentzian distance function and maximal curves, as well as the notion of `causal completeness' introduced in \cite{Gcambphil}, and used throughout \cite{horo1}. From Section \ref{secspheresandhoros} on, we assume that all spacetimes are globally hyperbolic.

\subsection{Lorentzian Distance and Maximal Curves}  \label{secdistance}

The Lorentzian arc length of a causal curve $\a : [a,b] \to M$ is defined by 
\beq
L(\a) := \int_a^b \sqrt{-g(\a',\a')}ds \nonumber
\eeq
It is a basic fact that causal geodesics are locally Lorentzian arc length-\emph{maximizing}. The Lorentzian distance function of $M$ is then defined by
\beq
d(p,q) := \sup \{L(\a) : \a \in \Omega^c_{p,q} \} \nonumber
\eeq
where $\Omega^c_{p,q}$ denotes the set of future causal curves from $p \in M$ to $q \in M$, and where we take the supremum to be zero if there are no such curves, i.e., if $p \not \le q$. 

A causal curve $\a$ is \emph{maximal} if it realizes the (Lorentzian) distance between any two of its points, i.e., $d(\a(s_1), \a(s_2)) = L(\a|_{[s_1,s_2]})$. A maximal curve is necessarily a timelike or null geodesic, (up to parameterization, a distinction which we will often ignore below).

We recall that a spacetime is globally hyperbolic if the set of all `timelike diamonds' $I^+(p) \cap I^-(q)$ forms a basis for the manifold topology, and all `causal diamonds' $J^+(p) \cap J^-(q)$ are compact. By a \emph{Cauchy surface} we mean an achronal set $S$ which is met by every inextendible causal curve in $M$. It is a basic fact that a spacetime is globally hyperbolic iff it admits a Cauchy surface, and that these conditions are related to Lorentzian distance as follows.

\begin{prop} \label{GHdist} Let $M$ be a spacetime and $d$ its Lorentzian distance function. If $M$ is globally hyperbolic, then $d$ is finite and continuous, and any causally related pair of points $p \le q$ are connected by a maximal causal geodesic $\a$, $L(\a) = d(p,q)$.
\end{prop}

\medskip
It is also natural to consider, for example, distance to the past of a subset $S \subset M$,
$$d(p, S) := \sup \{d(p,z) : z \in S\}$$
If $S$ is compact, Proposition \ref{GHdist} generalizes immediately. However, a natural, weaker compactness condition suffices, which we now review. As introduced in \cite{Gcambphil}, a subset $S \subset M$ is said to be \emph{past causally complete} if for all $p \in M$, the closure in $S$ of $J^+(p) \cap S$ is compact. It follows that such a set must be closed. Further, if $M$ is globally hyperbolic, then a closed set $S$ is past causally complete iff $J^+(p) \cap S$ is compact for all $p \in M$. (See Figure \ref{pcc}.)

\begin{figure} [h]
\begin{center}
\def\svgwidth{12cm} 
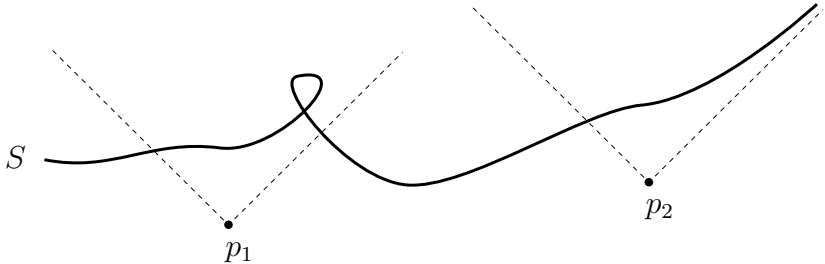
\caption[...]{Testing for past causal completeness.\label{pcc}} 
\end{center}
\end{figure}

Future causal completeness is defined time-dually. Any compact set is (both past and future) causally complete, as is any Cauchy surface. As shown in \cite{horo1}, we have the following generalization of Proposition \ref{GHdist}:

\begin{prop} \label{GHdistset} Let $C \subset M$ and consider the `past distance function'
\beq
d(x, C) = \sup \{d(x, z) : z \in C\} \nonumber
\eeq 
If $M$ is globally hyperbolic and $C$ is past causally complete, then $d(\cdot, C)$ is finite and continuous, and for every $p \in J^-(C)$, there is a maximal causal geodesic $\a$ from $p$ to $C$, with $L(\a) = d(p,C)$. Time-dual statements hold for the `future distance function', $d(C, x) := \sup \{d(z, x) : z \in C\}$, when $C$ is future causally complete.
\end{prop}

We close this section by reviewing rays and lines. A maximal, future causal curve $\a : [a, b) \to M$, with $a < b \le \infty$, which is future-inextendible is called a \emph{future ray}. \emph{Past rays} are defined time dually. By a \emph{line} we mean a maximal curve which is inextendible to both the future and past. Hence, rays and lines are necessarily causal geodesics, though they need not be complete as such. Finally, let $S \subset M$ be an arbitrary subset, and let $\a : [a,b) \to M$ be a future-inextendible causal curve, with $a < b \le \infty$, and $\a(a) \in S$. We say that $\a$ is a \emph{future $S$-ray} if $d(S, \a(t)) = L(\a_{[a,t]})$, for all $t \in [a, b)$. Past $S$-rays are defined time-dually. Note that $S$-rays are indeed rays as defined above. The following will be used below.

\begin{prop} \label{sray} If $S$ is a closed achronal $C^0$ hypersurface, then any null $S$-ray is contained in $S$. It follows that if $S$ is a Cauchy surface, then any $S$-ray is timelike. A compact Cauchy surface $S$ necessarily admits at least one future $S$-ray, and one past $S$-ray.
\end{prop}

\subsection{Lorentzian Spheres and Horospheres}  \label{secspheresandhoros}

Although we will continue to state this explicitly at times, we assume from now on that $M$ is globally hyperbolic.

Thinking of any subset $C \subset M$ as a `center', and fixing any radius $r > 0$, we can consider the corresponding \emph{past} and \emph{future spheres} from $C$, respectively,
$$S^-_r(C) := \{x \in M : d(x, C) = r\}$$
$$S^+_r(C) := \{x \in M : d(C, x) = r\}$$

\medskip
It is immediate that such spheres must be achronal, though edge points are possible in general. However, the following facts were established in \cite{horo1}. We recall first that a set $A \subset M$ is \emph{acausal} if no two points in $A$ are joined by a nontrivial causal curve.

\begin{prop} \label{spheresareABs} Let $M$ be globally hyperbolic, and  $C \subset M$ past causally complete. Then for any $r > 0$, the past sphere $S^-_r(C)$ is acausal and edgeless, with $S^-_r(C) = \partial I^-(S^-_r(C))$. Moreover, each point $x \in S^-_r(C)$ is connected to $C$ by a future timelike geodesic `radial' segment of length $r$. 
\end{prop}

Now consider a sequence of past spheres $\{S^-_k := S^-_{r_k}(C_k)\}$, with each $C_k$ past causally complete, and $r_k \to \infty$. If $\{S^-_k\}$ has an achronal limit as in Section \ref{secachlimits},
$$S^-_\infty := \lim \{S^-_k\}$$
then we call $S^-_\infty$ the \emph{past horosphere} associated to $\{S^-_k\}$. Time-dually, an analogous sequence of future spheres, $\{S^+_k\}$, leads to a \emph{future horosphere}, $S^+_\infty$.

We first note that the above represents a broad generalization of the horospheres in \cite{horo1}, by dropping the requirement that the sequence of `prehorospheres' be monotonic. At the same time, all of the basic properties of horospheres established in \cite{horo1} continue to hold, and indeed most of the proofs that remain after Section \ref{secachlimits} above carry over unchanged. In particular, combining Proposition \ref{spheresareABs}, Theorem \ref{achlimitsresults}, and Lemma \ref{Hausdorffseq} above, with Lemma 3.19 in \cite{horo1}, we have:

\begin{thm} \label{horosthm} Let $M$ be a globally hyperbolic spacetime, and let $S^-_\infty \ne \emptyset$ be a past horosphere as above. Then $S^-_\infty$ is an achronal boundary, and hence a closed, achronal $C^0$ hypersurface. Further, there is a `radial' future $S^-_\infty$-ray emanating from each point $x \in S^-_\infty$. If $S^-_\infty$ lies in the past of a Cauchy surface, then $S^-_\infty$ is acausal, and every future $S^-_\infty$-ray is timelike. Time-dual statements hold for future horospheres.
\end{thm}

We think of a past horosphere $S^-_\infty$ as `a past sphere centered at future infinity'. The future $S^-_\infty$-rays from each point are, roughly, the radial segments connecting $S^-_\infty$ to its `center', and arise precisely as limits of the radial segments of the corresponding sequence of spheres. Time-dually, a future horosphere $S^+_\infty$ can be thought of as a sphere centered at past infinity, with each point in $S^+_\infty$ connected to its `center' by a past $S^+_\infty$-ray. 

\medskip
\noindent
{\it Ray and Cauchy horospheres.}
We now briefly recall the two main horosphere constructions defined in \cite{horo1}. Although each remains unchanged from \cite{horo1}, this will serve both as concrete illustrations of the above, and also as preparation for Section \ref{secapps}, where both constructions will be used prominently. We begin with the construction of a `ray horosphere', which is closely related to the conventional Lorentzian Busemann horosphere. Let $\g : [0, \infty) \to M$ be a future complete timelike unit-speed ray. Taking the points $C_k := \g(k)$ as centers, it follows that the  past spheres $S^-_k(\g(k))$ form a monotonic sequence of achronal boundaries, with increasing pasts $P_k = I^-(S^-_k(\g(k)))$. Hence, as in Corollary \ref{oldachlimits}, we have a well-defined achronal limit $S^-_\infty(\g) := \lim \{S^-_k(\g(k))\}$, which we call the \emph{past ray horosphere} from $\g$. See \cite{horo1} for further details, and proofs of the following:

\begin{prop} \label{rayhoroprop} Suppose that $M$ is globally hyperbolic, and let $S^-_\infty(\g)$ be the past ray horosphere associated to a future complete timelike ray $\gamma$. Then $\g(0) \in S^-_\infty(\g) \subset \overline{I^-(\g)}$. In particular, $S^-_\infty(\g)$ is a nonempty, closed, achronal $C^0$ hypersurface, with future $S^-_\infty(\g)$-rays emanating from each point. In general, $\g$ is itself an $S^-_\infty(\g)$-ray. If $\g$ is a future $S$-ray for some Cauchy surface $S$, then $S^-_\infty(\g) \subset J^-(S)$, and hence $S^-_\infty(\g)$ is acausal and all future $S^-_\infty(\g)$-rays are timelike. 
\end{prop}

A `Cauchy horosphere', on the other hand, is built instead from a Cauchy surface $S$. We assume that $S$ is compact, and that $M$ is future timelike geodesically complete, so that the future spheres $C_k := S^+_k(S)$ are compact Cauchy surfaces as well, (see \cite{horo1}). Taking these as our sequence of centers, it follows that the sequence of past spheres $\{S^-_k(C_k)\}$ is monotonic, with decreasing pasts $P_k = I^-(S^-_k(C_k))$, and hence again we have a well-defined achronal limit $S^-_\infty(S) := \lim\{S^-_k(S^+_k(S))\}$, which we call the \emph{past Cauchy horosphere} from $S$.
The following is also established in  \cite{horo1}.

\begin{prop} \label{Cauchyhoroprop} Suppose that $M$ is future timelike geodesically complete, with compact Cauchy surface $S$, and let $S^-_\infty(S)$ be the past Cauchy horosphere from $S$. Then $\emptyset \ne S^-_\infty(S) \subset J^-(S)$. In particular, $S^-_\infty(S)$ is a nonempty, closed, acausal $C^0$ hypersurface, with future timelike $S^-_\infty(S)$-rays emanating from each point. In fact, letting $\g$ be any future $S$-ray, we have $\g(0) \in S^-_\infty(S)$, with $S^-_\infty(\g) \subset J^-(S^-_\infty(S))$.
\end{prop}

Finally, we note the time-dual constructions of the above. Namely, a past complete timelike ray $\b$ gives rise to a \emph{future ray horosphere}, $S^+_\infty(\b) : = \lim \{S^+_k(\b(k))\}$, and if $M$ is past timelike complete, we can similarly construct the \emph{future Cauchy horosphere}, $S^+_\infty(S) := \lim \{S^+_k(S^-_k(S))\}$, from any compact Cauchy surface $S$.

\subsection{Horosphere Structure and Rigidity} \label{sechorostructure}
We now present a new splitting result for general horospheres as defined in  Section \ref{secspheresandhoros}.  We begin with several key lemmas. We continue to assume throughout that $M$ is globally hyperbolic.

\begin{lem} \label{nullhoro} Let $S^-_\8$ be any past horosphere. Then $p \in S^-_\8$ is causally related to another point of $S^-_\8$ iff there is a null future $S^-_\8$-ray based at $p$.
\end{lem}

\begin{proof} Since a null $S^-_\8$-ray is necessarily contained in $S^-_\8$ (cf. Proposition \ref{sray}), one direction is trivial. Suppose then that $x, y \in S^-_\8$ with $x \le y$. Since $S^-_\8$ is achronal, $x$ and $y$ are necessarily joined by a (maximal) null geodesic segment (contained in $S^-_\8$). Being a past horosphere, $S^-_\8$ admits a future $S^-_\8$-ray $\b_y$ from $y$. If this were timelike, then by cutting the corner at $y$, we could produce a `longer' curve from $x$ to $\b_y$. Hence, $\b_y$ is null with $\b_y \subset S^-_\8$, and in fact, $\b_y$ must extend the null geodesic segment joining $x$ to $y$. Joining these then gives a future inextendible null $S^-_\8$-ray $\b_x$ from $x$, with $\b_y \subset \b_x \subset S^-_\8$. The statement of the lemma follows (with either $p = x$ or $p = y$).
\end{proof}

\begin{lem} \label{acausalhoro} Let $S^-_\8$ be any past horosphere and $p \in S^-_\8$. Then $S^-_\8$ fails to be acausal near $p$ iff there is a future null $S^-_\8$-ray from $p$.
\end{lem}

\begin{proof} To be precise, we say $S^-_\8$ is acausal near $p \in S^-_\8$ if there is some neighborhood $U$ of $p$ in $M$ such that $S^-_\8 \cap U$ is acausal in $M$. Suppose that $S^-_\8$ fails to be acausal near $p$. Fix a complete Riemannian metric $h$ on $M$ and let $U_k$ be the $h$-ball of radius $1/k$ around $p$. That $S^-_\8$ fails to be acausal near $p$ means that for each $k$, there are distinct points $x_k, y_k \in S^-_\8 \cap U_k$ with $x_k \le y_k$. Then, as in Lemma \ref{nullhoro}, there is a null $S^-_\8$-ray $\b_k : [0, \infty) \to M$ from $x_k$, parameterized with respect to $h$ arc length. Since $x_k \to p$, it follows, by standard arguments, that any limit curve $\b : [0, \8) \to M$ of $\{\b_k\}$ is a future $S^-_\8$-ray  from $p$. To see that $\b$ is null, one observes, for example, that for each $t > 0$, we have $d(\b(0), \b(t)) = \lim_{j \to \8} d(\b_{k_j}(0), \b_{k_j}(t)) = 0$. 
\end{proof}

Let $S^-_\infty$ be a past horosphere. Motivated by Lemmas \ref{nullhoro} and \ref{acausalhoro}, we will call $p \in S^-_\infty$ a \emph{null point} if there is a future null $S^-_\infty$-ray from $p$. Otherwise, we call $p$ a \emph{spacelike point} of $S^-_\infty$. Note that the set of spacelike points is open in $S^-_\infty$. The points of a future horosphere are classified time-dually. 

\begin{lem} \label{tangential} Suppose that $S^-_\8$ and $S^+_\8$ are past and future horospheres, respectively, satisfying $I^+(S^+_\8) \cap I^-(S^-_\8) = \emptyset$. Then at any intersection point $p \in S^-_\8 \cap S^+_\8$, one of the following situations holds:

\ben
\item [(1)] The point $p$ is a spacelike point for both horospheres, and there is a unique future $S^-_\infty$-ray from $p$, and a unique past $S^+_\infty$-ray from $p$, both of which are timelike and join to form a timelike line. 
\item [(2)] The point $p$ is a null point for both horospheres, and there is a unique future $S^-_\infty$-ray from $p$, and a unique past $S^+_\infty$-ray from $p$, both of which are null and join to form a null line, $\b$, with $\b \subset S^-_\infty \cap S^+_\infty$.
\een
\end{lem}

\begin{proof} Fix $p \in S^-_\8 \cap S^+_\8$. Then there is a future $S^-_\8$-ray $\g : [0, c) \to M$ from $p$, and a past $S^+_\8$-ray $\eta : [0, d) \to M$ from $p$, with $0 < c, d \le \infty$. Fixing any $0 < s < d$, and any $0 < t < c$, the initial segments of $\eta$ and $\g$ join to form a causal curve segment from $\eta(s)$ to $\g(t)$. Let $\s : [a,b] \to M$ be any other causal curve from $\s(a) = \eta(s)$ to $\s(b) = \g(t)$. Letting $P^\pm_\8$ and $F^\pm_\8$ be the unique past and future sets associated to $S^\pm_\8$ (as in Proposition  \ref{ABs}), we have $\eta(s) \in P^\pm_\8 \cup S^\pm_\8$ and $\g(t) \in S^\pm_\8 \cup F^\pm_\8$. Hence, $\s$ meets both horospheres. Let $\tau_-, \tau_+ \in [a,b]$ be any parameter times such that $\s(\tau_-) \in S^-_\8$ and $\s(\tau_+) \in S^+_\8$. Note that:
\beq
L(\s_{[a, \tau_+]}) \le d(\eta(s), S^+_\8) = L(\eta|_{[0,s]})  \nonumber
\eeq
\beq
L(\s_{[\tau_-, b]}) \le d(S^-_\8, \g(t)) = L(\g|_{[0,t]}) \nonumber
\eeq
Suppose first that $\tau_- \le \tau_+$. Then, by `counting $\s_{[\tau_-, \tau_+]}$ twice', we have:
\beq
L(\s) \le L(\s|_{[a,\tau_+]}) + L(\s|_{[\tau_-, b]}) \le L(\eta|_{[0,s]}) + L(\g_{[0,t]}) \nonumber\\
\eeq
Suppose now that $\tau_+ < \tau_-$. Then since $I^+(S^+_\8) \cap I^-(S^-_\8) = \emptyset$, the segment $\s|_{[\tau_+, \tau_-]}$ must be null, and does not contribute to the length of $\s$, and we have:
\beq
L(\s) = L(\s|_{[a,\tau_+]}) + \cancel{L(\s|_{[\tau_+, \tau_-]})} + L(\s|_{[\tau_-, b]}) \le L(\eta|_{[0,s]}) + L(\g_{[0,t]}) \nonumber
\eeq

It follows that $\eta$ and $\g$ join to form a line, and from this follows the uniqueness of $\eta$ and $\g$. In particular, either $\eta$ and $\g$ are both timelike, or they are both null. In the null case, $p$ is a null point for both horospheres. On the other hand, if $\eta$ and $\g$ are timelike, then by their uniqueness, $p$ must be a spacelike point for both horospheres.
What remains to show is that in the null case, $\b := - \eta + \g$ is contained in both horospheres. Since, in this case, $\g$ is a null $S^-_\8$-ray, we have $\g \subset S^-_\8$.  Moreover, if there is a point on $\g$ not in $S^+_\8$, then at some stage $\g$ must enter the timelike past of $S^+_\8$ (since by assumption it cannot enter the timelike future).  But this would violate the achronality of $S^+_\8$.  Hence, $\g \subset S^+_\8$. By a similar argument, $\eta$ is also contained in both horospheres.
\end{proof}

\begin{thm} \label{splitlem} Let $M$ be a globally hyperbolic, timelike geodesically complete spacetime, satisfying the timelike convergence condition, $\mathrm{Ric}(X,X) \ge 0$ for all timelike $X$. Let $S^-_\8$ be any past horosphere and $S^+_\8$ any future horosphere which meet at a common spacelike point $p \in S^-_\infty \cap S^+_\infty$, with $I^+(S^+_\8) \cap I^-(S^-_\8) = \emptyset$. Then $S^-_\8 = S^+_\8 =: S_\8$ is a smooth, geodesically complete spacelike Cauchy surface along which $M$ splits,
\beq
(M, g) \approx (\field{R} \times S_\8, -dt^2 + h) \nonumber
\eeq  
\end{thm}

\begin{proof} Recall that, by Lemma \ref{tangential}, any point in the intersection $S^-_\infty \cap S^+_\infty$ is either a spacelike point for both horospheres, or a null point for both horospheres. Let $U \subset S^-_\8 \cap S^+_\8$ be the  subset of spacelike intersection points. Hence, $p \in U$, and both $S^-_\8$ and $S^+_\8$ are acausal near any $x \in U$.   In particular, each future $S^-_\8$-ray $\a$ near such $x$ is timelike.  Parametrizing $\a : [0,\infty) \to M$ with respect to arc length, for any $r > 0$,  the past sphere $S^-_r(\a(r))$ is smooth near $\a(0)$ and lies locally to the past of $S^-_{\infty}$.  Using the timelike convergence condition, by standard comparison techniques (see e.g. \cite[Section 1.6]{Karcher}), $S^-_r(\a(r))$ has mean curvature $\ge -\frac{n}{r}$ at $\a(0)$.   It follows that $S^-_{\infty}$ has mean curvature $\ge 0$ in the support sense near any $x \in U$. Similarly, $S^+_\8$ has support mean curvature $\le 0$ near any such $x$. It then follows from the `support' maximum principle in \cite{AGH} that $U$ is a smooth, maximal spacelike hypersurface, which is open in both $S^-_\8$ and $S^+_\8$. But furthermore, since the normal geodesics from $U$, both to the future and the past, are timelike $U$-rays, which are complete (to the future or past) by assumption, standard Riccati (Raychaudhuri) equation techniques imply that $U$ is in fact totally geodesic, with split normal exponential image $(\mathrm{exp}^\perp(U), g) \approx (\field{R} \times U, -dt^2 + h)$, where $h$ is the induced metric on $U$.

We now extend to a \emph{global} splitting of all of $M$. As a first step, we show that $U$ is `geodesically closed', that is, any geodesic initially tangent to $U$ can never leave $U$. To that end, fix any $p \in U$ and any tangential vector $X \in T_pU$, and let $\s : (-a, b) \to M$ be the unique geodesic with $\s(0) = p$ and $\s'(0) = X$ which is maximally extended in $M$, where $0 < a,b \le \8$. Because $U$ is totally geodesic, $\s$ initially remains in $U$. Fix any $0 < s_0 < b$ with $\s([0, s_0)) \subset U$. Again, since $U$ is totally geodesic, i.e., its second fundamental form $K(X,Y) = g(\D_X N, Y)$ vanishes, the future unit normal field $N$ of $U$ is parallel along $\s|_{[0,s_0)}$. By Lemma \ref{tangential}, there is a unique future $S^-_\infty$-ray $\g_x$ from each $x \in U$, which is timelike. If we give each $\g_x$ a unit speed parameterization, then $\g_x'(0) = N_x$. Whether or not $q = \s(s_0)$ lies in $U$, there is a well-defined limit vector $N_q = \lim_{s \to s_0}N_{\s(s)}$, obtained by parallel transporting $N$ on all of $\s_{[0,s_0]}$, with $N_q$ necessarily future unit timelike. Let $\g_q$ be the future-directed unit speed timelike geodesic with $\g_q'(0) = N_q$, which is necessarily complete. Since $q \in \overline{U} \subset S^-_\8 \cap S^+_\8$, $\g_q$ is a curve from $S^-_\8$. Suppose that $\g_q|_{[0, \8)}$ is not an $S^-_\8$-ray, i.e., that for some $T > 0$, there is a point $z \in S^-_\8$ with $d(z, \g_q(T)) \ge T + 2 \e$, for some $\e > 0$. But then, for some neighborhood $W$ of $\g_q(T)$, we would have $d(z, w) \ge T + \e$, for all $w \in W$, which would contradict the fact that $\g_{\s(s)}$ is an $S^-_\8$-ray for all $s \in [0, s_0)$. Hence, $\g_q|_{[0,\8)}$ is a timelike future $S^-_\8$-ray. Since $q \in S^-_\8 \cap S^+_\8$, we have $q \in U$, (cf. Lemma \ref{tangential}). This shows that $\s$ can never leave $U$, i.e., we have 
$\s : (-a, b) \to U$. 

Now we show that, in fact, $\s$ must be complete. Without loss of generality, we take $\s$ to be unit speed. Suppose to the contrary that $b < \8$, for example. Then the curve $c(s) = (-2s, \s(s))$, $c : [0, b) \to \field{R} \times U \subset M$ is a past-directed timelike geodesic in $M$, and $\s(s) = \mathrm{exp}_{c(s)}(2s\d_t)$. By timelike completeness, $c$ extends to $[0,b]$. Furthermore, the vector field $\d_t$ is parallel in $\mathrm{exp}^\perp(U)$, and hence, by parallel translating along $c$, has a limit at $c(b)$. Hence, $\s(s) = \mathrm{exp}_{c(s)}(2s\d_t)$ has a limit as $s \to b$, i.e., $\s$ extends continuously, and hence as a geodesic to $[0, b]$. But this would contradict the definition of $b$. Thus, in fact, $b = \8$, and similarly $a = \8$, and $\s$ is complete. Since $\s$ was arbitrary, we have shown that $U$ is geodesically complete. But now a standard argument, using the product structure and geodesic completeness of $U$, shows that $J(U) = \mathrm{exp}^\perp(U) \approx \field{R} \times U$, and $H^\pm(U) = \emptyset$. Hence, $U$ is a Cauchy surface for $M$, (and is thus connected). This implies $S^-_\8 = U = S^+_\8$. 
\end{proof}

\smallskip
\noindent 
{\it Remark.}  One may ask, in the context of Theorem \ref{splitlem}, what, if any, rigidity occurs in the case that the common point $p \in S^-_\8 \cap S^+_\8$ is a null point.  In this case, by part (2) of Lemma \ref{tangential}, there is a null line passing through $p$, the future half  of which is contained  in  $S^-_\8$ and the past of which is contained in $S^+_\8$.  If $M$ is null geodesically complete and satisfies the null energy condition, $\ric(X,X) \ge 0$ for all null vectors $X$, then one can apply the results of \cite{GalNMP} to show that the components of $S^-_\8$ and $S^+_\8$ through $p$, respectively, agree along a totally geodesic null hypersurface.

\section{Applications} \label{secapps}

We now explore some applications of the framework developed above. In Section \ref{SecLST}, we show that the basic Lorentzian splitting theorem follows easily from Theorem \ref{splitlem} above. In Section \ref{SecBSC}, we give a new result on the Bartnik splitting conjecture. In Section \ref{Seccausalbdy}, we explore some connections between rigidity and the causal boundary, including some results in the case of positive cosmological constant in Section \ref{Secposcosmo}. First, however, we note that Theorem \ref{splitlem} generalizes the basic `$\Lambda = 0$' splitting result in \cite{horo1}, there labelled Theorem 4.4.

\begin{thm}  \label{horosplit} Let $M$ be a globally hyperbolic, timelike geodesically complete spacetime which satisfies the timelike convergence condition. Suppose that $S^-_\infty$ is a past horosphere which is future bounded, i.e., $S^-_\infty \subset J^-(\S)$ for some Cauchy surface $\S$. If $S^-_\infty$ admits a past $S^-_\infty$-ray, then $S^-_\infty$ is a smooth, spacelike geodesically complete Cauchy surface along which $M$ splits.
\end{thm} 

\begin{proof} To get this from Theorem \ref{splitlem}, let $\g$ be the past $S^-_\infty$-ray in the hypotheses. Since $S^-_\infty$ is future bounded, $S^-_\infty$ is acausal, and hence $\g$ must be timelike. Constructing the associated future ray horosphere $S^+_\infty(\g)$, we have $\g(0) \in S^-_\infty \cap S^+_\infty(\g)$, with $I^-(S^-_\infty) \cap I^+(S^+_\infty(\g)) = \emptyset$, and $\g(0)$ a spacelike point. 
\end{proof}

\subsection{The Lorentzian Splitting Theorem} \label{SecLST}

We now briefly note that Theorem \ref{splitlem} gives the basic Lorentzian splitting theorem (stated below) as an easy consequence. Let $\a : (-\8, \8) \to M$ be a complete future-directed unit speed timelike geodesic line. By the `\emph{future half}' of $\a$ we mean the future ray $\a^+ := \a|_{[0,\8)}$. By the `\emph{past half}' of $\a$ we mean the past ray $\a^- := - \a|_{(-\8, 0]}$. Denote the ray horospheres associated to each half of $\a$ by $S^-_\a := S^-_\8(\a^+)$ and $S^+_\a := S^+_\8(\a^-)$.

\begin{lem} \label{rayhoropair} The past and future pair of ray horospheres, $S^-_\a$ and $S^+_\a$, associated to each half of a complete timelike line $\a$ satisfy $I^+(S^+_\a) \cap I^-(S^-_\a) = \emptyset$.
\end{lem}

\begin{proof} Suppose otherwise that there are points $x \in S^-_\a$ and $y \in S^+_\a$ with $y \ll x$. Let $U$ be a neighborhood of $x$ and $V$ a neighborhood of $y$ such that, for all $u \in U$ and $v \in V$, we have $v \ll u$. Recall that $x$ is the limit of a sequence $x_k \in S^-_k(\a(k))$ and $y$ is the limit of a sequence $y_k \in S^+_k(\a(-k))$. Then, for $k_0$ a large enough integer so that both $x_{k_0} \in U$ and $y_{k_0} \in V$, we have $y_{k_0} \ll x_{k_0}$. But this leads to a contradiction of the maximality of $\a$.
\end{proof}

Since $\a^+$ is a future timelike $S^-_\a$-ray, (cf. Proposition \ref{rayhoroprop}), it then follows from Lemma \ref{tangential} that $\a(0)$ is a spacelike point for $S^-_\a$ and $S^+_\a$. Theorem \ref{splitlem} thus gives the following version of the Lorentzian splitting theorem:

\begin{thm} [Lorentzian Splitting Theorem] Let $M$ be a globally hyperbolic, timelike geodesically complete spacetime, satisfying $\mathrm{Ric}(X,X) \ge 0$ for all timelike $X$. If $M$ admits a timelike line $\a$, then $M$ splits. In particular, letting $S^-_\a$ and $S^+_\a$ be the ray horospheres associated to each half of $\a$, then $S^-_\a = S^+_\a =: S_\a$ is a smooth, spacelike, geodesically complete Cauchy surface for $M$ and $(M,g) \approx (\field{R} \times S_\a, -dt^2 + h)$, where $h$ is the induced metric on $S_\a$. 
\end{thm}

\subsection{The Bartnik Splitting Conjecture} \label{SecBSC}

The problem of establishing a Lorentzian splitting theorem, posed by Yau in the early 80's, was in fact originally motivated by the question of rigidity in the classical singularity theorems of Hawking and Penrose. The ultimate resolution of the splitting theorem did not, however, settle this rigidity question. In \cite{Bart88}, Bartnik realized this question concretely as follows.

\begin{conj} [Bartnik Splitting Conjecture, `88] \label{Bartnikconj} Suppose that $M$ is a globally hyperbolic spacetime, with compact Cauchy surfaces, which satisfies the timelike convergence condition, $\textrm{Ric}(X,X) \ge 0$ for all timelike $X$. If $M$ is timelike geodesically complete, then $M$ splits as $(M,g) \approx (\field{R} \times \S, -dt^2 + h)$, where $\S$ is a smooth spacelike Cauchy hypersurface, with induced metric $h$.
\end{conj}

In physical terms, the conjecture roughly states that in a spatially closed, relativistic spacetime (with $\Lambda = 0$), any dynamics whatsoever will always lead to singularities. The conjecture is illustrated mathematically by the warped product case, $g = -dt^2 + f^2(t)h$, for which the timelike convergence condition implies $f'' \le 0$. 

The Bartnik conjecture has been shown to hold under various auxiliary conditions. (See for example, \cite{Bart88}, \cite{EschGal}, \cite{GalBanach}.) To our knowledge, the weakest of these include the `ray-to-ray' condition in \cite{GalBanach}, and the `max-min' condition in \cite{horo1}. While a direct comparison of these two conditions may not be obvious, we will give a condition below which is weaker than both, and under which Conjecture \ref{Bartnikconj} still holds.

In \cite{GalBanach}, the first author established the following:

\begin{thm} [\cite{GalBanach}] \label{raytoraythm} Let $M$ be a (future or past) timelike geodesically complete spacetime with compact Cauchy surface $S$. Suppose that there is a future $S$-ray $\g$ and a past $S$-ray $\eta$ such that $I^+(\eta) \cap I^-(\g) \ne \emptyset$. Then $M$ admits a timelike line. 
\end{thm}

The result above appears as Theorem 4.4 in \cite{GalBanach}, with the assumption of full timelike completeness. We note, however, that timelike completeness in either direction suffices, as the proof only involves applying Lemma 4.2 in \cite{GalBanach} in one direction. (See also Lemma 3.15 in \cite{horo1}.)

In fact, it is straightforward to see that the `ray-to-ray' condition in Theorem \ref{raytoraythm} can be weakened so that the future and past rays may be from different Cauchy surfaces. That is, the construction of the timelike line given in \cite{GalBanach} still goes through if we only assume that there are (compact) Cauchy surfaces $S$ and $\S$, and a future $S$-ray $\g$ and past $\S$-ray $\eta$ satisfying $I^+(\eta) \cap I^-(\g) \ne \emptyset$. In particular, this generalized ray-to-ray condition is sufficient to give the splitting in Conjecture \ref{Bartnikconj}.

A different condition was explored in \cite{horo1}. Suppose that $M$ is  future timelike geodesically complete. Let $S$ be a compact Cauchy surface and set $S_k := \{x \in M : d(S, x) = k\}$, that is, $S_k = S^+_k(S)$. Then each $S_k$ is itself a compact Cauchy surface. Set $M_k := \mathrm{max} \{ d(x, S_k) : x \in S\}$, and $m_k := \mathrm{min} \{ d(x, S_k) : x \in S\}$. We note that $M_k = k$, but we may have $m_k < k$. Then, as in \cite{horo1}, we say the \emph{`max-min' condition} holds on $S$ if, for some positive constant $C > 0$, we have $M_k - m_k \le C$, for all $k$. The basic practical implication of this condition is that if $S$ satisfies the max-min condition, and if $M$ is, in addition, past timelike complete, then the past Cauchy horosphere $S^-_\infty(S)$ is compact. Compactness of any horosphere is of consequence in a variety of ways, especially in the context of splitting. In particular, if $S^-_\infty(S)$ is compact, it follows that $S^-_\infty(S)$ is a Cauchy surface. Hence, $S^-_\infty(S)$ is `future-bounded' by itself, and thus under the hypotheses of the Bartnik conjecture, Theorem \ref{horosplit} applies to split $M$. This result, which we now state formally, appears as Theorem 4.9 in \cite{horo1}.

\begin{thm} [\cite{horo1}] \label{maxminthm} Let $M$ be a timelike geodesically complete spacetime which satisfies the timelike convergence condition. It $S$ is a compact Cauchy surface which satisfies the max-min condition, then its past Cauchy horosphere $S^-_\infty(S)$ is a smooth compact spacelike Cauchy surface, along which $M$ splits.
\end{thm}

Hence, the Bartnik splitting conjecture holds under the additional assumption of either the (generalized) `ray-to-ray' condition, or the `max-min' condition. We now consider a kind of `horo-to-horo' condition, which, in the context of the conjecture, is implied by either of these, but still sufficient to give the splitting.

\begin{lem} \label{horotohorolem} Suppose that $M$ is timelike geodesically complete, with compact Cauchy surfaces. If either the (generalized) ray-to-ray condition, or the max-min condition holds, then there are two Cauchy surfaces $S$ and $\S$ such that $J^+(S^+_\8(\S)) \cap J^-(S^-_\8(S)) \ne \emptyset$.
\end{lem}

\begin{proof} Suppose first that the generalized ray-to-ray condition holds, that is, that there is a Cauchy surface $S_0$ with future $S_0$-ray $\g : [0,\infty) \to M$, and a Cauchy surface $\S_0$ with past $\S_0$-ray $\eta : [0,\infty) \to M$, with both rays parameterized with respect to arc length, such that $I^+(\eta) \cap I^-(\g) \ne \emptyset$. Hence, we have $\eta(a) \ll \g(b)$, for some $0 \le a, b < \infty$. Letting $S := S^+_b(S_0)$, then $S$ is a (compact) Cauchy surface, and the tail $\g|_{[b, \infty)}$ is a future $S$-ray. Constructing the past Cauchy horosphere $S^-_\infty(S)$, we have $\g(b) \in S^-_\infty(S) \cap S$, as in Proposition \ref{Cauchyhoroprop}. Time-dually, letting $\S := S^-_a(\S_0)$, then $\S$ is a Cauchy surface, with $\eta(a) \in S^+_\infty(\S) \cap \S$. In particular, this shows that $J^+(S^+_\8(\S)) \cap J^-(S^-_\8(S)) \ne \emptyset$. 

Now suppose instead that $S$ is a Cauchy surface which satisfies the max-min condition. It follows that the past Cauchy horosphere $S^-_\infty(S)$ is a compact Cauchy surface, (see \cite{horo1} for details). Letting $\S := S^-_\infty(S)$, then the future Cauchy horosphere $S^+_\infty(\S)$ has a point in common with $\S$. That is, we have $S^+_\infty(\S) \cap S^-_\infty(S) \ne \emptyset$, from which the conclusion follows trivially. 
\end{proof}

We now show the Bartnik conjecture holds under the condition in Lemma \ref{horotohorolem}.

\begin{thm} \label{horotohorothm} Let $M$ be a timelike geodesically complete spacetime, with compact Cauchy surfaces, which satisfies the timelike convergence condition. Suppose that there are two Cauchy surfaces $S$ and $\S$ such that $J^+(S^+_\8(\S)) \cap J^-(S^-_\8(S)) \ne \emptyset$. Then $M$ splits.
\end{thm}

\begin{proof} We recall first that the two Cauchy horospheres $S^-_\8(S)$ and $S^+_\8(\S)$ are acausal, (cf. Proposition \ref{Cauchyhoroprop}), and hence, in particular, consist entirely of spacelike points. Now, note that the nontrivial intersection occurs within the compact region $J^+(\S) \cap J^-(S)$. Hence, there is a finite distance $d := \max \{d(p,q) : p \in S^+_\infty(\S), q \in S^-_\infty(S)\}$, and points $p_0 \in S^+_\infty(\S)$ and $q_0 \in S^-_\infty(S)$, with $d(p_0, q_0) = d(S^+_\infty(\S), S^-_\infty(S)) = d$. If $d = 0$, then it follows (along the lines of Proposition \ref{sray}) that the two horospheres `meet without crossing' as in Theorem \ref{splitlem}, and we are done. If $d > 0$, the basic idea of the proof is to replace $S^-_\infty(S)$ with a modified horosphere which again gives the situation of Theorem \ref{splitlem}. More precisely, suppose that $d > 0$. Recall that $S^-_\infty(S) = \lim \{S^-_k(S^+_k(S))\}$. We then consider the modified horosphere $S^-_{\infty+d}(S) := \lim \{S^-_{k+d}(S^+_k(S))\}$, which we claim is the same as taking the past sphere from the original Cauchy horosphere, i.e., $S^-_{\infty + d}(S) = S^-_d(S^-_\infty(S))$. (See Figure \ref{horo2horo}.) Suppose for now that this holds. Since $d(p_0, S^-_\infty(S)) = d$, we have $p_0 \in S^-_d(S^-_\infty(S)) = S^-_{\infty + d}(S)$. Hence, $S^+_\infty(\S) \cap S^-_{\infty + d}(S) \ne \emptyset$. Furthermore, using the fact that $S^-_{\infty + d}(S) = S^-_d(S^-_\infty(S))$, it follows that we must have $I^+(S^+_\infty(\S)) \cap I^-(S^-_{\infty + d}(S)) = \emptyset$.  Hence, Theorem \ref{splitlem} gives the splitting.

\begin{figure}
\begin{center}
\def\svgwidth{13cm} 
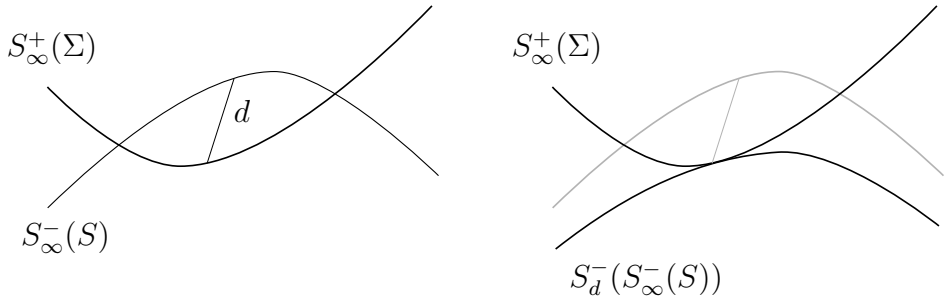
\caption[...]{The past sphere from $S^-_\infty(S)$ gives a new horosphere.}
\label{horo2horo}
\end{center}
\end{figure}

It remains to show that $S^-_{\infty+d}(S) := \lim \{S^-_{k+d}(S^+_k(S))\}$ exists, and equals $S^-_d(S^-_\infty(S))$.  With regard to the latter, we note that, as was shown in \cite{horo1}, $S^-_\infty(S)$ is past causally complete, and hence is an appropriate center for a past sphere. To simplify notation, let $C_k := S^+_k(S)$. Hence, $S^-_\infty(S) = \lim\{S^-_k(C_k)\}$ and $S^-_{\infty + d}(S) = \lim \{S^-_{k+d}(C_k)\}$. To show that the latter (exists and) equals the past sphere from the former, it suffices to show the two inclusions: (1) $\overline{\lim}\{S^-_{k+d}(C_k)\} \subset S^-_d(S^-_\infty(S))$ and (2) $S^-_d(S^-_\infty(S)) \subset \underline{\lim}\{S^-_{k+d}(C_k)\}$. The proof of these inclusions is facilitated by the fact that, as derived in \cite{horo1}, we have $S^-_k(C^-_k) \subset J^-(S)$ for all $k$, and that, moreover,
$S^-_\infty(S)$ is a `past achronal limit' as defined in \cite{horo1}, (cf.\ the last statement in Proposition 2.5 in \cite{horo1}). We leave the details to the interested reader.
\end{proof}

The proof of Theorem \ref{horotohorothm} above is a natural and immediate application of Theorem \ref{splitlem}, (cf. Figure \ref{horo2horo}). However, we note also that the case $I^+(S^+_\infty(\S)) \cap I^-(S^-_\infty(S)) \ne \emptyset$ can alternatively be obtained from Theorem 3.7 in \cite{GalBanach}.

\subsection{Connections to the causal boundary} \label{Seccausalbdy}

We conclude with an exploration of some connections between rigidity and the causal boundary of spacetime. More specifically, we consider the case that the (past or future) causal boundary of $M$ is spacelike. In Section \ref{SecBartspacelikebdy}, we show in particular that the Bartnik splitting conjecture holds in this case. In Section \ref{Secposcosmo}, we explore this condition in the setting of positive cosmological constant.

We begin with some  comments about the causal boundary of a globally hyperbolic spacetime $(M,g)$, cf. \cite{GKP, HE} for further details.  We shall focus on the past causal boundary $\mathscr{C}^-$; time-dual statements apply to the future causal boundary $\mathscr{C}^+$.  Heuristically, $\mathscr{C}^-$ consists of `ideal points' which represent the `past end points' of past inextendible timelike curves.  This is made precise in terms of indecomposable future sets (IFs).  Let $F$ be a future set, $F = I^+(S)$ for some set $S \subset M$.  Recall,  $F$ is a future set if and only if  
$I^+(F) = F$.  By definition, $F$ is an indecomposable future set if it cannot be expressed as the union of two future sets which are proper subsets of $F$.  It can be shown  \cite{HE} that there are only two types of IFs: the timelike future of a point $p$, $I^+(p)$, and the timelike future of a past inextendible timelike curve $\g$, $I^+(\g)$.  The latter sets are called terminal indecomposable future sets, or TIFs for short.  The \emph{past causal boundary} $\mathscr{C}^-$ is, by definition, the set of all TIFs (with obvious identifications).   Following the terminology of Wald and Yip \cite{WaldYip},  $\mathscr{C}^-$ is said to be \emph{spacelike} if no TIF  is properly contained in another.

\subsubsection{Ray-to-ray and the causal boundary} \label{SecBartspacelikebdy}

There is a connection between the ray-to-ray condition and the causal boundary of spacetime, which yields yet another `special case' of the Bartnik splitting conjecture.

\begin{prop}\label{r2r}  Let $(M,g)$ be a spacetime with compact Cauchy surfaces.  If the past causal boundary of $(M,g)$ is spacelike then $(M,g)$ satisfies the ray-to-ray condition.  Hence, if in addition $(M,g)$ is future timelike goedesically complete, $(M,g)$ contains a timelike line.
\end{prop}

This proposition is consistent with the example constructed in \cite{EhrGal}, which does not contain any timelike lines, and whose past causal boundary is $\mathscr{C^-}$ is nontimelike, but not spacelike. 

\begin{proof} [Proof of Proposition \ref{r2r}] Let $S$ be a compact Cauchy surface for $M$, and let $\g$ be a future $S$-ray starting at $p \in S$.  Let $\s$ be a past inextendible timelike curve starting at $p$.  Use $\s$ to construct a past $S$-ray in the usual manner:  Take a sequence of points $p_k$ along $\s$ that exhaust $\s$ to the past, and, for each $k$, construct a past directed maximizer to $S$,
$\eta_k$, from $S$ to $p_k$.   Since $S$ is compact, by a standard limit curve argument, a subsequence of the $\eta_k$'s converges to a past $S$-ray $\eta$, which, again by compactness, must be timelike.    By construction, $\eta \subset \overline{I^+(\s)}$, and hence by standard properties, $I^+(\eta) \subset  I^+(\s)$.  Since the past causal boundary is assumed to be spacelike, we must have $I^+(\eta) =  I^+(\s)$.  Hence, $I^-(\g) \cap I^+(\eta) = I^-(\g) \cap I^+(\s) \ne \emptyset$.
\end{proof}

Proposition \ref{r2r} thus gives the Bartnik conjecture under the additional assumption of spacelike boundary:

\begin{thm}\label{Bartspacelikebdythm}
Let $(M,g)$ be a spacetime which contains a
compact Cauchy surface and satisfies the timelike convergence condition.  Assume that either the future or past causal boundary of 
$M$ is spacelike.  If $(M,g)$ is timelike geodesically complete,
then $(M,g)$ splits as in the Bartnik conjecture.
\end{thm}

Indeed, Proposition  \ref{r2r}, or its time-dual, implies that $M$ admits a timelike line. One can then apply the Lorentzian splitting theorem.  We note that the past causal boundary of a complete product spacetime  with compact Cauchy surfaces, as in the conclusion of the theorem, is trivially spacelike, since it consists of a single TIF. Indeed, the product structure and compactness of the Cauchy surfaces implies that $I^+(\g) = M$ for any past inextendible timelike curve, and hence  $\mathscr{C}^-$ consists of a single element.

It is natural to ask if the timelike convergence condition, or a curvature condition consistent with this (involving only weak inequalities), could be used to show, in the context of Proposition \ref{r2r}, that the past causal boundary is necessarily spacelike.  Recall \cite{BEE} that the spacetime Ricci curvature tensor evaluated on a unit timelike vector can be expressed as {\it minus} the sum of timelike sectional curvatures.  Theorem 3 in \cite{EhrGal}, which is based on a causality theorem of Harris \cite{Harris}, can be used to show the following. 

\begin{prop}\label{sectional} Let $(M, g)$ be a spacetime with compact Cauchy surfaces and with
everywhere non-positive timelike sectional curvatures, $K \le 0$.   If $(M, g )$ is past
timelike geodesically complete then the past causal boundary $\mathscr{C}^-$ is spacelike; in fact 
$\mathscr{C}^-$ consists of single element.  
\end{prop}

\begin{proof} Indeed, one has $I^+(\g) = M$ for any past inextendible timelike curve.  For if this were not the case, then there would be a past inextendible timelike curve $\g$ such that $\d I^+(\g) \ne \emptyset$.  By properties of achronal boundaries \cite{Penrose}, $\d I^+(\g)$ is an achronal $C^0$ hypersurface ruled by past inextendible null geodesics.  However, by the time-dual of \cite[Theorem 3]{EhrGal}, any such null geodesic would enter its own timelike past, thereby violating the achronality of $\d I^+(\g)$.
\end{proof}

Of course, \cite[Theorem 3]{EhrGal} also shows that there can be no null lines in such a spacetime, and hence  the standard causal line construction must give rise to a timelike line.  

\subsubsection{An application with positive cosmological constant} \label{Secposcosmo}

In this section we consider spacetimes $(M^{n+1},g)$ which obey the 
Einstein equations,
\beq\label{einstein}
R_{ij} -\frac12Rg_{ij} +\Lambda g_{ij} = 8\pi T_{ij} \,, 
\eeq
with  positive cosmological constant $\Lambda$, where the energy-momentum tensor
$T_{ij}$  is assumed to satisfy  the  strong energy condition, 
\beq
(T_{ij} - \frac1{n-1}Tg_{ij})X^iX^j \ge 0 \label{SEC}
\eeq
for all timelike vectors $X$, where $T = T_i{}^i$.

After a rescaling, we may assume $\Lambda = n(n-1)/2$.
With this normalization the strong energy condition is equivalent to
\beq\label{EC}
\ric (X,X) \ge -n  \quad \text{for all unit  timelike vectors } X.
\eeq

The aim of this section is to prove the following singularity theorem and a rigidity result from which it follows.

\begin{thm}\label{FRWsing}  Let $(M,g)$ be a globally hyperbolic spacetime satisfying:
\vs{-.07in}
\ben 
\item[(1)] $(M,g)$ obeys \eqref{einstein}-\eqref{SEC}, with $\Lambda = n(n-1)/2$.

\vs{-.07in}
\item[(2)] $(M,g)$  has spacelike past  causal boundary $\mathscr{C}^-$.

\vs{-.07in}
\item[(3)] $(M,g)$ admits a noncompact geodesically complete spacelike Cauchy surface $V$ of nonpositive scalar curvature, $S \le 0$.

\vs{-.07in}
\item[(4)] The local energy density along $V$ is nonnegative, $\mu := T(u,u) = T_{ij}u^iu^j \ge 0$, where $u$ is the future directed unit normal to $V$.
\een
Finally, assume $V$  has nonnegative mean curvature at some point. Then $(M,g)$ is past timelike geodesically incomplete; in fact, some timelike geodesic orthogonal to $V$ is past incomplete.  

\end{thm}

A distinctive feature of this theorem (in addition to the assumption of a spacelike past causal boundary) is that the Cauchy surface is required to be noncompact.  Theorem \ref{FRWsing} is well-illustrated by the classical dust-filled FLRW models satisfying \eqref{einstein} with $\Lambda > 0$: see, e.g., \cite[chapter 23]{dinverno}.  The spatially isotropic Cauchy surfaces in these models are, up to a time-dependent scale factor, complete simply connected spaces of constant (sectional) curvature $k = +1, 0, -1$.  If the so-called `mass parameter' is sufficiently small, the `closed' models ($k = +1$) will be past timelike geodesically complete (the limiting case being that of de Sitter space).  However, the `open' models ($k = 0, -1$), to which, in fact, our theorem applies, are all past timelike geodesically incomplete, and in fact all begin with a big bang singularity.

\smallskip
Theorem \ref{FRWsing} is a simple consequence of the following theorem.

\begin{thm}\label{warpsplit}
Let $(M,g)$ be a globally hyperbolic spacetime which satisfies the energy condition \eqref{EC}, and which has a spacelike past causal boundary $\mathscr{C}^-$.    Suppose $M$ admits a geodesically complete spacelike Cauchy surface $V$ with mean curvature $H \ge n$.  If all timelike geodesics orthogonal to $V$ are past complete then $V$ is necessarily compact, and 
$(J^-(V), g)$ is isometric to the warp product  $([0,\infty) \times V, -dt^2 + e^{-2t} h$), where $h$ is the induced metric on 
$V$. (Here $\frac{\d}{\d t}$ is past pointing.) 
\end{thm}

Theorem \ref{warpsplit} may be viewed as an extension of Proposition 3.4 in \cite{AndGal}, to the case of a complete, but not necessarily compact Cauchy surface.  

Suppose that $(M,g)$ and $V$ satisfy the hypotheses of Theorem \ref{FRWsing}.  Contraction of the Gauss equation for $V \subset M$ leads to the Hamiltonian constraint,
\beq
S - 2\Lambda  - |K|^2 + H^2   = 16\pi \mu   \,, \nonumber
\eeq 
where $K$ is the second fundamental form of $V$. Using $\mu \ge 0$, 
$\Lambda = n(n-1)/2$, $S\le 0$,  and  $|K|^2 \ge H^2/n$ (by the Cauchy-Schwartz inequality) in the  above gives,
\beq
H^2  \ge n^2 -\frac{n}{n-1} S  \ge n^2 \,. \nonumber
\eeq
Since $H$ is assumed to be nonnegative somewhere, we conclude that $H \ge n$.  It follows that  $(M,g)$ and $V$ satisfy the conditions of Theorem \ref{warpsplit}.  Since $V$ is assumed to be noncompact, we see that the conclusion of  Theorem \ref{FRWsing} now follows from Theorem~\ref{warpsplit}.

We now focus attention on the proof of Theorem \ref{warpsplit}.  The proof makes essential use of the following result of Wald and Yip \cite{WaldYip}.

\begin{lem}[\cite{WaldYip}]\label{wald} 
Let $(M,g)$ be a  spacetime with Cauchy surface $S$ and with spacelike past causal boundary $\mathscr{C}^-$.  Then for any TIF $W$, $S \cap \overline{W}$ is compact.
\end{lem}

This is the key consequence of assuming the past causal boundary is spacelike.   It will also be convenient for the proof of Theorem \ref{warpsplit} to single out the following lemma.

\begin{lem}\label{lemsplit}
Let $(M,g)$ be a globally hyperbolic spacetime satisfying the energy condition \eqref{EC}.   Let $V$ be a spacelike Cauchy surface for $M$ with mean curvature $H \ge n$, such that each timelike geodesic orthogonal to $V$ is past complete.  If each such geodesic is a $V$-ray then  $(J^-(V), g)$ is isometric to the warp product  $([0,\infty) \times V, -dt^2 + e^{-2t} h$), where $h$ is the induced metric on $V$.
\end{lem}

\begin{proof} [Proof of Lemma \ref{lemsplit}]  
The proof technique is fairly standard.  The $V$-ray assumption implies  that there can be no focal points to $V$ along any past directed normal geodesic.  Furthermore, no two past directed normal geodesics can intersect.  It follows that global Gaussian normal coordinates can be introduced on $J^-(V)$, i.e. up to isometry, $J^-(V) = [0,\infty) \times V$, and on $J^-(V)$, 
\beq\label{metric} 
g = -dt^2 + h_t  \,, 
\eeq 
where, for each $t \in [0,\infty)$, $h_t$ is the induced metric on $V_t = \{t\} \times V$; in local coordinates we write, $h_t = h_{ij}(t,x) dx^i dx^j$.

For each $t$, let $H = H_t$ and $K = K_t$ be the mean curvature and second fundamental form, respectively of $V_t$ defined with respect to the {\it future} unit normal  field $u = -\frac{\d}{\d t}$.  $H = H_t$ obeys the traced Riccati (Raychaudhuri) equation,
\beq\label{riccati}
\frac{\partial H}{\partial t} = \ric(u,u) + \frac{H^2}{n} + |\s|^2  \,, 
\eeq
where $\s = \s_t$, the {\it shear tensor}, is the trace free part of $K$, $\s = K - \frac{H}{n} h$.   Since, by assumption 
$\ric(u,u) \ge -n$, we obtain the differential inequality
\beq
\frac{\partial H}{\partial t} \ge \frac{H^2}{n} - n \,, \nonumber
\eeq
where, in addition, $H(0) \ge n$.  It follows by basic comparison techniques (see e.g., \cite[Section 1.6]{Karcher}, 
\cite[Proposition  3.4]{AndGal}) that $H(t) \ge n$ for all $t \ge 0$.  If moreover, one had $H(t) > n$ at some point along a normal  geodesic to $V$, then, by the same sort of comparison techniques, $H$ would necessarily diverge to $+\infty$ in finite parameter time $t$, which would be a contradiction.    Hence, we must in fact have $H(t) \equiv n$ for all $t \ge 0$.  Using this and the Ricci curvature condition \eqref{EC} in \eqref{riccati}, we conclude that the shear $\s$ vanishes for all $t \ge 0$,  which in turn implies that  $K_t = h_t$, or in terms of coordinates, $K_{ij} = h_{ij}$ for all $t \ge 0$.   Since, in our Gaussian normal coordinates, $K_{ij} = -\frac12\frac{\partial h_{ij}}{\partial t}$, we obtain 
$h_t = e^{-2t}h$.  Insertion of this in \eqref{metric} yields the desired result.
\end{proof}

\begin{proof} [Proof of Theorem \ref{warpsplit}]   We use Lemma \ref{wald} to construct a past $V$-ray.  Let $\mu$ be a past inextendible timelike curve starting on $V$.   Consider a sequence of points $q_k$ on $\mu$, $q_{k+1} \in I^-(q_k)$, exhausting 
$\mu$ to the past.  Let $\g_k$ be a past directed timelike geodesic maximizer to $V$, from $p_k \in V$ to $q_k$; $\g_k$ meets $V$ orthogonally.   Since we are assuming the past causal boundary is spacelike,  Lemma \ref{wald} applied to the TIF $I^+(\mu)$ gives that $V \cap \overline{I^+(\mu)}$ is compact.  Since each $p_k \in V \cap \overline{I^+(\mu)}$, it follows that a subsequence of the $\g_k$'s converges to a past inextendible timelike geodesic $\g$ orthogonal $V$ starting at $p \in V$, say.  By the maximality of the $\g_k$'s , $\g$ is a past complete $V$-ray. 

Now, let  $S^+_{\infty} = S^+_{\infty}(\g)$ be the future ray horosphere associated to $\g$.   By Proposition \ref{rayhoroprop}, we know (i) $S^+_{\infty} \subset J^+(V)$, (ii) $S^+_{\infty}$ passes through $p = \g(0)$, and 
(iii) $S^+_{\infty}$ is acausal and there is a past timelike $S^+_{\infty}$-ray emanating from each point.  Consider any such $S^+_{\infty}$-ray $\a$. By the manner in which these rays are constructed (see \cite[Lemma 3.22]{horo1}), $\a \subset \overline{I^+(\g)}$, and hence $I^+(\a) \subset I^+(\g)$.  Since we are assuming the past causal boundary is spacelike, we must in fact have $I^+(\a) = I^+(\g)$. It follows that for any point on $\g$ there is a point on $\a$ in its timelike past. Hence, since $\g$ is past complete and $\a$ maximizes length to $S^+_{\infty}$, $\a$ must also be past complete. Let $\a : [0,\infty) \to M$ be parameterized with respect to arc length.  For any $r > 0$,  the future sphere $S^+_r(\a(r))$ is smooth near $\a(0)$ and lies locally to the future of $S^+_{\infty}$.  Using~\eqref{EC},  by standard comparison techniques, (\cite[Section 1.6]{Karcher}, \cite[Lemma 6.4]{horo1}) $S^+_r(\a(r))$ has mean curvature $\le n\coth r$ at $\a(0)$. It follows that $S^+_{\infty}$ has mean curvature $\le n$ in the support sense.  Let $S^+$ be the connected component of $S^+_\infty$ which contains $\g(0)$. Then $S^+ \cap V$ is non-empty and closed. Since $S^+$ meets $V$ locally to the future near any intersection point $x \in S^+\cap V$, the maximum principle in \cite{AGH} gives that, for some spacetime neighborhood $U$ of $x$, we have $V \cap U =  S^+ \cap U$. It follows that $S^+ \cap V$ is open in both $V$ and $S^+$, and hence that $V= S^+$. Consequently, the timelike past $S^+_\infty$-rays from each point of $S^+ = V$ are also $V$-rays. But these $V$-rays are precisely the past normal geodesics from $V$. 

We may now apply Lemma \ref{lemsplit} to conclude that $(J^-(V), g)$ is isometric to the warped product  $([0,\infty) \times V, -dt^2 + e^{-2t} h$), where $h$ is the induced metric on $V$.   To complete the proof, we show that if $V$ is noncompact then  the past causal boundary of $(J^-(V), g)$ (which agrees with the past causal boundary of $(M,g)$) is not spacelike, contrary to assumption.  Under the change of variable $u = e^t -1$, $g$ on $J^-(V) = \{(u,x): u \ge0, x \in V\}$ becomes $g = (u+1)^{-2}\tilde g$, where $\tilde g$ is the product metric, $\tilde g = -du^2 +h$.   Since the causal boundary is conformally invariant, we may work with $\tilde g$.  In what follows, all spacetime quantities refer to this metric. If $(V,h)$ is complete and noncompact, from any point $q \in V$, we can construct a unit speed ray $\s: [0, \infty) \to V$, $s \to \s(s)$.  Then $\eta: [0, \infty) \to J^-(S)$, defined by $\eta(s) = (s, \s(s))$ is a past inextendible achronal null geodesic. By considering a past inextendible timelike curve in $I^+(\eta)$ that asymptotes to $\eta$, we see that $I^+(\eta)$ defines a TIF.  Let $\b$ be the past directed time line, $\beta(u) = (u,q)$, $u \ge 0$.   By the product structure of $\tilde g$,  $\eta \subset I^+(\b)$, and hence $I^+(\eta) \subset I^+(\b)$. On the other hand, since $\eta$ is achronal, $q \notin I^+(\eta)$. Thus, $I^+(\eta)$ is a proper subset of $I^+(\b)$, which implies that the past causal boundary is not spacelike.  Hence, $V$ must be compact.
\end{proof}

\bibliographystyle{amsplain}
\bibliography{horosequel_arXiv}

\end{document}